\documentclass{article}

    \PassOptionsToPackage{numbers, compress}{natbib}

 \usepackage[preprint]{neurips_2025}

\usepackage[utf8]{inputenc} %
\usepackage[T1]{fontenc}    %
\usepackage{hyperref}       %
\usepackage{url}            %
\usepackage{booktabs}       %
\usepackage{amsfonts}       %
\usepackage{nicefrac}       %
\usepackage{microtype}      %
\usepackage{xcolor}         %
\usepackage{graphicx}
\usepackage{amsmath, amssymb, amsthm}
\usepackage{caption}  
\usepackage{multirow}
\usepackage[table]{xcolor}
\usepackage{tikz}
\usepackage{enumitem}
\usetikzlibrary{matrix, positioning, arrows.meta, calc}

\theoremstyle{definition}
\newtheorem{definition}{Definition}[section]
\newtheorem{theorem}{Theorem}[section]
\newtheorem{lemma}{Lemma}[section]

\newcommand{\R}{\mathbb{R}}
\newcommand{\N}{\mathbb{N}}
\newcommand{\supp}{\text{supp}}
\newcommand{\vu}{\mathbf{u}}
\newcommand{\vv}{\mathbf{v}}
\newcommand{\vm}{\mathbf{m}}
\newcommand{\vt}{\mathbf{t}}
\newcommand{\vc}{\mathbf{c}}
\newcommand{\ve}{\mathbf{e}}
\newcommand{\vx}{\mathbf{x}}
\newcommand{\vy}{\mathbf{y}}
\newcommand{\vf}{\mathbf{f}}

\newcommand{\vz}{\mathbf{0}}

\newcommand{\vd}{\mathbf{d}}
\newcommand{\cU}{\mathcal{U}}
\newcommand{\cV}{\mathcal{V}}
\newcommand{\cQ}{\mathcal{Q}}
\newcommand{\cD}{\mathcal{D}}

\DeclareMathOperator{\sgn}{sgn}
\DeclareMathOperator*{\argmax}{arg\,max}

\title{Quantifying and Expanding the Theoretical Capacity of Late-Interaction Retrieval Models}

\author{%
  Julian Killingback\\
  Center for Intelligent Information Retrieval\\
  University of Massachusetts Amherst\\
  \texttt{jkillingback@cs.umass.edu}\\
  \And
  Varad Ingale\\
  Center for Intelligent Information Retrieval\\
  University of Massachusetts Amherst\\
  \texttt{vingale@umass.edu} \\
  \And
  Hamed Zamani\\
  Center for Intelligent Information Retrieval\\
  University of Massachusetts Amherst\\
  \texttt{zamani@cs.umass.edu} \\
  \And
  Cameron Musco\\
  University of Massachusetts Amherst\\
  \texttt{cmusco@cs.umass.edu} \\
}

\begin{document}

\maketitle

\begin{abstract}
Late-interaction retrieval models that use the MaxSim similarity function have shown strong empirical performance, often outperforming single-vector dense and sparse retrieval models. Despite these empirical findings, little is known about the \textit{theoretical} representation power of MaxSim and how it compares to other retrieval approaches. This paper shows by construction that MaxSim similarity can exactly replicate the inner product between any two non-negative $k$-sparse vectors with possibly infinite dimension, requiring only $O(k)$ representation space. Moreover, there exist similarities that MaxSim can express while standard vector inner products with the same representation space cannot. Leveraging our theoretical framework, we introduce \textit{Signed MaxSim} which allows late-interaction models to exactly replicate any real-valued inner product, something we prove standard MaxSim is not capable of. We also show that MaxSim can act as an aggregation of soft-OR operations and as an evaluator of logical expressions in positive Conjunctive Normal Form. Our findings show that MaxSim is at least as capable as standard vector inner products for any non-negative vectors and our extension, Signed MaxSim, is as capable for any vectors. Both similarities possess additional capabilities that inner product cannot replicate, marking one of the first theoretical justifications and quantifications of late-interaction methods. 
Our theoretical findings are supported empirically: on a retrieval task featuring queries with negations, Signed MaxSim improves out-of-domain performance significantly over a standard ColBERT/MaxSim baseline with nDCG@10 increasing from 0.597 to 1.000 under a vocabulary shift and from 0.008 to 0.788 on negation-only queries.
\end{abstract}

\section{Introduction}
The landscape of neural information retrieval is currently dominated by two primary modeling paradigms. The first, comprising both dense retrievers (e.g., DPR \citep{dpr}) and learned sparse retrievers (e.g., SNRM \citep{SNRM}), encodes queries and documents into single, fixed-dimensional vectors. Despite their differences in sparsity, both rely on a simple inner product to estimate relevance. The second paradigm, exemplified by late-interaction models like ColBERT \citep{colbert}, represents texts as sets of embeddings and estimates relevance with more complex similarity measures, most often MaxSim (Chamfer Similarity)--a sum of maximum similarities between query and document embeddings.

\begin{definition}[MaxSim Similarity]
\label{def:maxsim}
The MaxSim similarity $S: \cU \times \cV \to \R$, where $\cU, \cV \subset \R^{n}$, is defined as:
\[
       S(\mathcal{U}, \mathcal{V}) = \sum_{\mathbf{m} \in \mathcal{U}} \max_{\mathbf{t} \in \mathcal{V}} \langle\mathbf{m}, \mathbf{t} \rangle.
\]
Here, $ \langle \cdot, \cdot \rangle$ is the inner product.
\end{definition}

Generally speaking, late-interaction models have empirically demonstrated superior performance, particularly in out-of-domain settings \citep{dense_on_blog}, but the underlying mechanism for this gap has not been fully elucidated. Empirical work has shown that this gap is not explained by the additional representation space leveraged by late-interaction models \citep{colbert, embedding_dim_scaling_laws}, as providing more embedding dimensions to single-vector approaches has diminishing returns and does not match the performance of late-interaction models. These results suggest that it is the MaxSim similarity that is the main differentiator. In this work, we theoretically prove this hypothesis, showing that MaxSim similarity can exactly replicate the inner product between any non-negative vectors and that there are similarities that MaxSim can produce with a given representation size which an inner product cannot reproduce. These findings show that MaxSim is more capable than inner product in some cases and at least as capable as inner product when both vectors are non-negative. With that said, to show complete parity with inner-product-based retrieval methods, it must be possible to replicate inner products between any real-valued vectors, not just non-negative ones. We prove that this is impossible for standard MaxSim with a fixed embedding dimension and the number of embeddings tied to the support of the original vectors. However, this is not a fundamental limit of late-interaction retrievers, as we demonstrate with our extension to MaxSim, made possible by our theoretical framework, that enables the exact replication of the inner product between any real-valued vectors.

These findings explain why late-interaction models can outperform inner-product-based retrieval methods: they can leverage a more capable similarity function. To understand how this more powerful similarity function can be leveraged, we show how MaxSim relates to both sparse inner products and Boolean logic evaluation, two of the oldest building blocks in Information Retrieval. In detail, our three main findings are:

\textbf{1. MaxSim Subsumes Inner-Product Similarity.} We prove by construction that the MaxSim operation over sets of vectors in $\R^3$ can exactly reconstruct the inner product of any two non-negative vectors from a high-dimensional space. This result holds for both sparse and dense vectors. 
Specifically, the exact inner product between two $k$-sparse vectors can be achieved using exactly $k$ query embeddings and $k+1$ document embeddings in $\R^3$ representing each of the sparse vectors. This guarantees that late-interaction models inherently possess the full representational capacity of any standard non-negative single-vector retriever. This shows that MaxSim can replicate many existing retrieval approaches, such as learned-sparse approaches (e.g. SNRM \citep{SNRM} and SPLADE \citep{splade}) and traditional term matching methods (e.g. BM25 \citep{bm25}). However, the non-negative requirement makes it hard to compare with retrieval approaches, like dense retrieval, that use arbitrary real-valued vectors. Additionally, supporting negative values has clear utility for several query types. Negation queries, for example, naturally benefit from being able to directly reduce the score of documents that contain the negated concept. To address this limitation, we introduce an extension to the standard late-interaction representations and MaxSim similarity that allows for the exact reconstruction of inner products between arbitrary real-valued vectors (illustrated in Figure \ref{fig:maxsim_full_comparison_stacked}).  We also prove that, under mild assumptions, standard MaxSim cannot replicate this ability. Our modification is thus at least as capable as any inner-product-based or MaxSim-based retrieval method and enables better performance for certain categories of queries such as those with negations. This finding additionally relates MaxSim to sparse inner product, which have been core building blocks of information retrieval systems for decades.

\textbf{2. Separating MaxSim from Single-Vector Neural Retrieval.} We prove separation between the two paradigms regarding high-dimensional sparse spaces. We prove that no single finite-dimensional vector embedding can preserve the inner products of $k$-sparse vectors drawn from an arbitrarily high-dimensional ambient space. In contrast, MaxSim can achieve this exact preservation using sets of vectors in $\R^3$ with size $k$ by our previously discussed result. This mathematically formalizes why late-interaction models are uniquely suited for handling the ``long tail'' of vocabulary terms and rare entities that single-vector models inevitably compress or lose.

\textbf{3. Logical Expressivity.} Beyond comparing to inner product-based approaches, we analyze MaxSim as a way to evaluate logical expressions. We show that it can naturally function as an aggregation of Soft-ORs and that this allows MaxSim to act as a rank-equivalent Conjunctive Normal Form (i.e. an AND of ORs) evaluator for expressions without negations. This makes it ideal for matching synonyms and multiple forms of a relevant concept without overly rewarding the presence of multiple surface forms in a document. It also connects MaxSim to the classic Boolean search approaches used in early information retrieval systems and suggests that MaxSim may enable a neural approach which can function in a similar way to prior systems that required manual creation of the logical expressions used to query retrieval systems.

Together, our results suggest that the success of late-interaction models is not due only to the additional representation space, but also to a similarity measure  that is at least as expressive as the inner product, capable of infinite-dimensional sparse representation, and natively aligned with Boolean logic.

The rest of this paper is organized into five main sections. In Section \ref{sec:late_interaction_for_exact_inner_product}, we prove that MaxSim can exactly reproduce inner products between sparse non-negative vectors with possibly infinite dimension and show that our extension to MaxSim can exactly replicate the inner product of any sparse real-valued vectors. In Section \ref{sec:dim_bottleneck}, we prove that, unlike MaxSim, the inner product between two finite vectors cannot replicate an inner product between sparse infinite-dimensional vectors, thus demonstrating a separation between MaxSim and single-vector retrievers. These results show that late-interaction retrieval models are uniquely capable and can be viewed as a generalization of both dense and sparse single-vector retrieval models. Section \ref{sec:logical_evaluation} shows that MaxSim with larger embedding dimensions can be seen as an aggregation of weighted fuzzy OR operations and that by selecting specific values it can be viewed as evaluating positive Conjunctive Normal Form (CNF) expressions so that a ranked set of documents is ordered in an identical way to an exact Boolean evaluation. Finally, in Sections \ref{sec:experiments} and  \ref{sec:results}, we cover our experimental procedure and results comparing our MaxSim extension against standard MaxSim on queries with negations. We find that there is a significant improvement when evaluated in-domain and even greater improvement on out-of-domain data, which illustrates the additional robustness that is provided by our extension.

 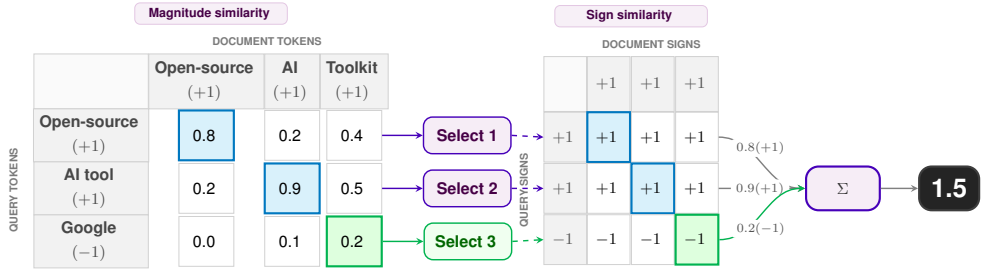
\begin{figure}[htbp]
\centering
\resizebox{0.94\textwidth}{!}{
\begin{tikzpicture}[
    >=stealth,
    font=\sffamily,
    sectiontitle/.style={
        font=\sffamily\Large\bfseries,
        text=black!88,
        anchor=west
    },
    doclabel/.style={
        font=\sffamily\normalsize\bfseries,
        text=black!82,
        anchor=west
    },
    statusbadge/.style={
        font=\sffamily\scriptsize\bfseries,
        text=black!58,
        fill=black!3,
        draw=black!12,
        rounded corners=3pt,
        inner xsep=5pt,
        inner ysep=2.5pt
    },
    matrixname/.style={
        font=\sffamily\scriptsize\bfseries,
        text=violet!55!black,
        fill=violet!6,
        draw=violet!22,
        rounded corners=3pt,
        minimum width=3.0cm,
        minimum height=0.42cm,
        inner xsep=7pt,
        inner ysep=2.5pt,
        align=center
    },
    axislabel/.style={
        font=\sffamily\tiny\bfseries,
        text=black!52,
        fill=white,
        inner xsep=3pt,
        inner ysep=1.5pt,
        align=center
    },
    cell/.style={
        draw=black!18,
        fill=white,
        minimum width=1.08cm,
        minimum height=0.82cm,
        font=\sffamily\small,
        align=center,
        anchor=center
    },
    head/.style={
        cell,
        fill=black!5,
        text=black!72,
        font=\sffamily\small\bfseries,
        minimum height=0.92cm
    },
    rowhead/.style={
        cell,
        fill=black!5,
        text=black!72,
        font=\sffamily\small\bfseries,
        minimum width=2.25cm
    },
    corner/.style={
        cell,
        fill=black!2,
        minimum width=2.25cm,
        minimum height=0.92cm
    },
    magcell/.style={
        cell,
        minimum height=0.98cm
    },
    maghead/.style={
        head,
        minimum height=1.08cm
    },
    magrow/.style={
        rowhead,
        minimum height=0.98cm
    },
    magcorner/.style={
        corner,
        minimum height=1.08cm
    },
    signcell/.style={
        cell,
        minimum width=0.84cm,
        minimum height=0.98cm,
        font=\sffamily\small\bfseries
    },
    signhead/.style={
        signcell,
        fill=black!5,
        text=black!68,
        minimum height=1.08cm
    },
    signrow/.style={
        signcell,
        fill=black!5,
        text=black!68,
        minimum width=0.84cm,
        minimum height=0.98cm
    },
    signcorner/.style={
        signcell,
        fill=black!2,
        minimum width=0.84cm,
        minimum height=1.08cm
    },
    posmax/.style={
        fill=cyan!13,
        draw=cyan!55!blue,
        very thick
    },
    negmax/.style={
        fill=orange!18,
        draw=orange!85!red,
        very thick
    },
    goodmax/.style={
        fill=green!13,
        draw=green!55!cyan!70!black,
        very thick
    },
    route/.style={
        draw=violet!55!blue,
        fill=violet!7,
        text=violet!55!black,
        thick,
        rounded corners=5pt,
        minimum width=1.72cm,
        minimum height=0.72cm,
        inner xsep=6pt,
        font=\sffamily\small\bfseries,
        align=center
    },
    op/.style={
        draw=cyan!55!blue,
        fill=cyan!6,
        text=black!78,
        thick,
        rounded corners=5pt,
        minimum width=1.72cm,
        minimum height=0.72cm,
        inner xsep=7pt,
        font=\sffamily\small\bfseries,
        align=center
    },
    sumop/.style={
        draw=violet!55!blue,
        fill=violet!7,
        text=black!80,
        very thick,
        rounded corners=7pt,
        minimum width=1.45cm,
        minimum height=0.86cm,
        font=\sffamily\small\bfseries,
        align=center
    },
    result/.style={
        draw=black!76,
        fill=black!86,
        text=white,
        very thick,
        rounded corners=7pt,
        minimum width=1.18cm,
        minimum height=0.86cm,
        font=\sffamily\Large\bfseries,
        align=center
    },
    scorelabel/.style={
        font=\sffamily\tiny\bfseries,
        text=black!52,
        fill=white,
        inner xsep=3pt,
        inner ysep=1.5pt,
        align=center
    },
    flow/.style={
        ->,
        very thick,
        draw=black!46
    },
    selectflow/.style={
        ->,
        thick,
        draw=violet!55!blue
    },
    scoreflow/.style={
        ->,
        thick,
        draw=black!50
    },
    term/.style={
        font=\sffamily\scriptsize\bfseries,
        fill=white,
        text=black!68,
        inner xsep=3pt,
        inner ysep=1.5pt,
        rounded corners=2pt
    },
    sectionline/.style={
        draw=black!16,
        line width=0.9pt
    }
]

\def\MatrixLabels#1#2#3#4{
    \coordinate (#1-top-center) at
        ($(#1-1-1.north west)!0.5!(#1-1-4.north east)$);

    \coordinate (#1-column-center) at
        ($(#1-1-2.north west)!0.5!(#1-1-4.north east)$);

    \coordinate (#1-row-center) at
        ($(#1-2-1.west)!0.5!(#1-4-1.west)$);

    \node[
        matrixname,
        anchor=south
    ] at ($(#1-top-center)+(0,0.57)$) {
        #2
    };

    \node[
        axislabel,
        anchor=south
    ] at ($(#1-column-center)+(0,0.10)$) {
        #3
    };

    \node[
        axislabel,
        rotate=90
    ] at ($(#1-row-center)+(-0.38,0)$) {
        #4
    };
}

\node[sectiontitle] (titleA) at (0,0)
    {A. Standard MaxSim};

\node[statusbadge, right=0.35cm of titleA]
    {CANNOT DISTINGUISH THE DOCUMENTS};

\node[
    doclabel,
    below=0.62cm of titleA.west,
    anchor=west
] (doc1_std_lbl) {
    Doc 1
    \textcolor{orange!80!red}{(irrelevant)}:
    open-source AI tool that mentions ``Google''
};

\matrix (M1) [
    matrix of nodes,
    nodes in empty cells,
    row sep=-\pgflinewidth,
    column sep=-\pgflinewidth,
    below=1.04cm of doc1_std_lbl.south west,
    anchor=north west,
    ampersand replacement=\&
] {
    |[corner]| {}
        \& |[head]| {\shortstack{Open-\\source}}
        \& |[head]| AI
        \& |[head]| Google
        \\
    |[rowhead]| Open-source
        \& |[cell, posmax]| 0.8
        \& |[cell]| 0.2
        \& |[cell]| 0.1
        \\
    |[rowhead]| AI tool
        \& |[cell]| 0.2
        \& |[cell, posmax]| 0.9
        \& |[cell]| 0.3
        \\
    |[rowhead]| No Google
        \& |[cell, negmax]| 0.4
        \& |[cell]| 0.1
        \& |[cell]| -0.5
        \\
};

\MatrixLabels
    {M1}
    {Raw similarity}
    {DOCUMENT TOKENS}
    {QUERY TOKENS}

\node[op, right=1.55cm of M1-2-4] (max1_1)
    {Max\;0.8};

\node[op, right=1.55cm of M1-3-4] (max1_2)
    {Max\;0.9};

\node[
    op,
    draw=orange!75!red,
    fill=orange!7,
    right=1.55cm of M1-4-4
] (max1_3) {
    Max\;0.4
};

\draw[flow]
    (M1-2-4.east) -- (max1_1.west);

\draw[flow]
    (M1-3-4.east) -- (max1_2.west);

\draw[flow, draw=orange!75!red]
    (M1-4-4.east) -- (max1_3.west);

\node[sumop] (sum1) at (15.85,0 |- max1_2)
    {$\Sigma$};

\coordinate (res1pos) at (18.75,0 |- max1_2);

\node[result, anchor=east] (res1) at (res1pos)
    {2.1};

\node[scorelabel, anchor=south]
    at ($(res1.north)+(0,0.14)$)
    {FINAL SCORE};

\draw[scoreflow]
    (max1_1.east)
    to[out=0, in=180]
    (sum1.west);

\draw[scoreflow]
    (max1_2.east) -- (sum1.west);

\draw[scoreflow]
    (max1_3.east)
    to[out=0, in=180]
    (sum1.west);

\draw[scoreflow]
    (sum1.east) -- (res1.west);

\node[
    doclabel,
    below=0.88cm of M1.south west,
    anchor=north west
] (doc2_std_lbl) {
    Doc 2
    \textcolor{green!55!cyan!70!black}{(relevant)}:
    open-source AI toolkit with no Google mention
};

\matrix (M2) [
    matrix of nodes,
    nodes in empty cells,
    row sep=-\pgflinewidth,
    column sep=-\pgflinewidth,
    below=1.04cm of doc2_std_lbl.south west,
    anchor=north west,
    ampersand replacement=\&
] {
    |[corner]| {}
        \& |[head]| {\shortstack{Open-\\source}}
        \& |[head]| AI
        \& |[head]| Toolkit
        \\
    |[rowhead]| Open-source
        \& |[cell, posmax]| 0.8
        \& |[cell]| 0.2
        \& |[cell]| 0.4
        \\
    |[rowhead]| AI tool
        \& |[cell]| 0.2
        \& |[cell, posmax]| 0.9
        \& |[cell]| 0.5
        \\
    |[rowhead]| No Google
        \& |[cell, goodmax]| 0.4
        \& |[cell]| 0.1
        \& |[cell]| 0.2
        \\
};

\MatrixLabels
    {M2}
    {Raw similarity}
    {DOCUMENT TOKENS}
    {QUERY TOKENS}

\node[op, right=1.55cm of M2-2-4] (max2_1)
    {Max\;0.8};

\node[op, right=1.55cm of M2-3-4] (max2_2)
    {Max\;0.9};

\node[
    op,
    draw=green!55!cyan!70!black,
    fill=green!6,
    right=1.55cm of M2-4-4
] (max2_3) {
    Max\;0.4
};

\draw[flow]
    (M2-2-4.east) -- (max2_1.west);

\draw[flow]
    (M2-3-4.east) -- (max2_2.west);

\draw[flow, draw=green!55!cyan!70!black]
    (M2-4-4.east) -- (max2_3.west);

\node[sumop] (sum2) at (15.85,0 |- max2_2)
    {$\Sigma$};

\coordinate (res2pos) at (18.75,0 |- max2_2);

\node[result, anchor=east] (res2) at (res2pos)
    {2.1};

\draw[scoreflow]
    (max2_1.east)
    to[out=0, in=180]
    (sum2.west);

\draw[scoreflow]
    (max2_2.east) -- (sum2.west);

\draw[scoreflow]
    (max2_3.east)
    to[out=0, in=180]
    (sum2.west);

\draw[scoreflow]
    (sum2.east) -- (res2.west);

\draw[sectionline]
    ($(M2.south west)+(0,-1.05)$)
    --
    ($(res2.east |- M2.south west)+(0,-1.05)$);

\node[
    sectiontitle,
    below=1.82cm of M2.south west,
    anchor=north west
] (titleB) {
    B. Signed MaxSim $S_{\pm}$
};

\node[statusbadge, right=0.35cm of titleB]
    {SEPARATES ROUTING FROM SIGN};

\node[
    doclabel,
    below=0.62cm of titleB.west,
    anchor=west
] (doc1_sgn_lbl) {
    Doc 1
    \textcolor{orange!80!red}{(irrelevant)}:
    open-source AI tool that mentions ``Google''
};

\matrix (Mag1) [
    matrix of nodes,
    nodes in empty cells,
    row sep=-\pgflinewidth,
    column sep=-\pgflinewidth,
    below=1.04cm of doc1_sgn_lbl.south west,
    anchor=north west,
    ampersand replacement=\&
] {
    |[magcorner]| {}
        \& |[maghead]| {\shortstack{\strut Open-source\\\strut $(+1)$}}
        \& |[maghead]| {\shortstack{\strut AI\\\strut $(+1)$}}
        \& |[maghead]| {\shortstack{\strut Google\\\strut $(+1)$}}
        \\
    |[magrow]| {\shortstack{\strut Open-source\\\strut $(+1)$}}
        \& |[magcell, posmax]| 0.8
        \& |[magcell]| 0.2
        \& |[magcell]| 0.1
        \\
    |[magrow]| {\shortstack{\strut AI tool\\\strut $(+1)$}}
        \& |[magcell]| 0.2
        \& |[magcell, posmax]| 0.9
        \& |[magcell]| 0.3
        \\
    |[magrow]| {\shortstack{\strut Google\\\strut $(-1)$}}
        \& |[magcell]| 0.0
        \& |[magcell]| 0.1
        \& |[magcell, negmax]| 0.7
        \\
};

\MatrixLabels
    {Mag1}
    {Magnitude similarity}
    {DOCUMENT TOKENS}
    {QUERY TOKENS}

\node[route, right=0.8cm of Mag1-2-4] (route1_1)
    {Select 1};

\node[route, right=0.8cm of Mag1-3-4] (route1_2)
    {Select 2};

\node[route, right=0.8cm of Mag1-4-4] (route1_3)
    {Select 3};

\draw[selectflow]
    (Mag1-2-4.east) -- (route1_1.west);

\draw[selectflow]
    (Mag1-3-4.east) -- (route1_2.west);

\draw[selectflow, draw=orange!75!red]
    (Mag1-4-4.east) -- (route1_3.west);

\matrix (Sign1) [
    matrix of nodes,
    nodes in empty cells,
    row sep=-\pgflinewidth,
    column sep=-\pgflinewidth,
    right=2.8cm of Mag1,
    anchor=west,
    ampersand replacement=\&
] {
    |[signcorner]| {}
        \& |[signhead]| $+1$
        \& |[signhead]| $+1$
        \& |[signhead]| $+1$
        \\
    |[signrow]| $+1$
        \& |[signcell, posmax]| $+1$
        \& |[signcell]| $+1$
        \& |[signcell]| $+1$
        \\
    |[signrow]| $+1$
        \& |[signcell]| $+1$
        \& |[signcell, posmax]| $+1$
        \& |[signcell]| $+1$
        \\
    |[signrow]| $-1$
        \& |[signcell]| $-1$
        \& |[signcell]| $-1$
        \& |[signcell, negmax]| $-1$
        \\
};

\MatrixLabels
    {Sign1}
    {Sign similarity}
    {DOCUMENT SIGNS}
    {QUERY SIGNS}

\draw[selectflow, dashed]
    (route1_1.east) -- (Sign1-2-1.west);

\draw[selectflow, dashed]
    (route1_2.east) -- (Sign1-3-1.west);

\draw[
    selectflow,
    dashed,
    draw=orange!75!red
] (route1_3.east) -- (Sign1-4-1.west);

\node[sumop] (sum3) at (16.05,0 |- Sign1-3-4)
    {$\Sigma$};

\coordinate (res3pos) at (18.75,0 |- Sign1-3-4);

\node[result, anchor=east] (res3) at (res3pos)
    {1.0};

\node[scorelabel, anchor=south]
    at ($(res3.north)+(0,0.14)$)
    {FINAL SCORE};

\draw[scoreflow]
    (Sign1-2-4.east)
    to[out=0, in=180]
    node[
        term,
        pos=0.48,
        above=2pt
    ] {$0.8(+1)$}
    (sum3.west);

\draw[scoreflow]
    (Sign1-3-4.east)
    --
    node[
        term,
        pos=0.48
    ] {$0.9(+1)$}
    (sum3.west);

\draw[
    scoreflow,
    draw=orange!75!red
]
    (Sign1-4-4.east)
    to[out=0, in=180]
    node[
        term,
        pos=0.48,
        below=2pt
    ] {$0.7(-1)$}
    (sum3.west);

\draw[scoreflow]
    (sum3.east) -- (res3.west);

\node[
    doclabel,
    below=0.98cm of Mag1.south west,
    anchor=north west
] (doc2_sgn_lbl) {
    Doc 2
    \textcolor{green!55!cyan!70!black}{(relevant)}:
    open-source AI toolkit with no Google mention
};

\matrix (Mag2) [
    matrix of nodes,
    nodes in empty cells,
    row sep=-\pgflinewidth,
    column sep=-\pgflinewidth,
    below=1.04cm of doc2_sgn_lbl.south west,
    anchor=north west,
    ampersand replacement=\&
] {
    |[magcorner]| {}
        \& |[maghead]| {\shortstack{\strut Open-source\\\strut $(+1)$}}
        \& |[maghead]| {\shortstack{\strut AI\\\strut $(+1)$}}
        \& |[maghead]| {\shortstack{\strut Toolkit\\\strut $(+1)$}}
        \\
    |[magrow]| {\shortstack{\strut Open-source\\\strut $(+1)$}}
        \& |[magcell, posmax]| 0.8
        \& |[magcell]| 0.2
        \& |[magcell]| 0.4
        \\
    |[magrow]| {\shortstack{\strut AI tool\\\strut $(+1)$}}
        \& |[magcell]| 0.2
        \& |[magcell, posmax]| 0.9
        \& |[magcell]| 0.5
        \\
    |[magrow]| {\shortstack{\strut Google\\\strut $(-1)$}}
        \& |[magcell]| 0.0
        \& |[magcell]| 0.1
        \& |[magcell, goodmax]| 0.2
        \\
};

\MatrixLabels
    {Mag2}
    {Magnitude similarity}
    {DOCUMENT TOKENS}
    {QUERY TOKENS}

\node[route, right=0.8cm of Mag2-2-4] (route2_1)
    {Select 1};

\node[route, right=0.8cm of Mag2-3-4] (route2_2)
    {Select 2};

\node[
    route,
    draw=green!55!cyan!70!black,
    fill=green!6,
    text=green!40!black,
    right=0.8cm of Mag2-4-4
] (route2_3) {
    Select 3
};

\draw[selectflow]
    (Mag2-2-4.east) -- (route2_1.west);

\draw[selectflow]
    (Mag2-3-4.east) -- (route2_2.west);

\draw[
    selectflow,
    draw=green!55!cyan!70!black
] (Mag2-4-4.east) -- (route2_3.west);

\matrix (Sign2) [
    matrix of nodes,
    nodes in empty cells,
    row sep=-\pgflinewidth,
    column sep=-\pgflinewidth,
    right=2.8cm of Mag2,
    anchor=west,
    ampersand replacement=\&
] {
    |[signcorner]| {}
        \& |[signhead]| $+1$
        \& |[signhead]| $+1$
        \& |[signhead]| $+1$
        \\
    |[signrow]| $+1$
        \& |[signcell, posmax]| $+1$
        \& |[signcell]| $+1$
        \& |[signcell]| $+1$
        \\
    |[signrow]| $+1$
        \& |[signcell]| $+1$
        \& |[signcell, posmax]| $+1$
        \& |[signcell]| $+1$
        \\
    |[signrow]| $-1$
        \& |[signcell]| $-1$
        \& |[signcell]| $-1$
        \& |[signcell, goodmax]| $-1$
        \\
};

\MatrixLabels
    {Sign2}
    {Sign similarity}
    {DOCUMENT SIGNS}
    {QUERY SIGNS}

\draw[selectflow, dashed]
    (route2_1.east) -- (Sign2-2-1.west);

\draw[selectflow, dashed]
    (route2_2.east) -- (Sign2-3-1.west);

\draw[
    selectflow,
    dashed,
    draw=green!55!cyan!70!black
] (route2_3.east) -- (Sign2-4-1.west);

\node[sumop] (sum4) at (16.05,0 |- Sign2-3-4)
    {$\Sigma$};

\coordinate (res4pos) at (18.75,0 |- Sign2-3-4);

\node[result, anchor=east] (res4) at (res4pos)
    {1.5};

\draw[scoreflow]
    (Sign2-2-4.east)
    to[out=0, in=180]
    node[
        term,
        pos=0.48,
        above=2pt
    ] {$0.8(+1)$}
    (sum4.west);

\draw[scoreflow]
    (Sign2-3-4.east)
    --
    node[
        term,
        pos=0.48
    ] {$0.9(+1)$}
    (sum4.west);

\draw[
    scoreflow,
    draw=green!55!cyan!70!black
]
    (Sign2-4-4.east)
    to[out=0, in=180]
    node[
        term,
        pos=0.48,
        below=2pt
    ] {$0.2(-1)$}
    (sum4.west);

\draw[scoreflow]
    (sum4.east) -- (res4.west);

\end{tikzpicture}
}
\caption{
Comparison of scoring for the query
\textit{``Open-source AI tools that do not mention Google.''}
Both documents satisfy the positive request for an open-source AI tool,
but only Doc~2 satisfies the exclusion constraint.
\textbf{(A)} Standard MaxSim selects the largest raw similarity for each
query token. Although the similarity to the explicit ``Google'' token
can be negative, the maximum operation instead selects the incidental
similarity of $0.4$ to ``Open-source'' for both documents. The two
documents therefore receive the same final score, so standard MaxSim
cannot use the exclusion constraint to distinguish the relevant document.
\textbf{(B)} Signed MaxSim separately uses magnitude similarity to identify
the matching feature and sign similarity to determine its contribution.
The explicit occurrence of ``Google'' therefore produces a large negative
contribution, while the document with no Google mention receives only a
small incidental penalty and is correctly ranked first.
}
\label{fig:maxsim_full_comparison_stacked}
\end{figure}

\section{Related Work}
In this section, we explore prior work related to late-interaction retrieval models and approaches to theoretically understand the capabilities of retrieval models.

\subsection{Late-Interaction Retrieval Models}
Late-interaction models have become a prominent class of retrieval models thanks to their strong empirical performance. The category was created with the introduction of ColBERT \citep{colbert}, which uses the MaxSim similarity (Chamfer similarity) to find the similarity between queries and documents. Although the empirical results are strong, the necessity to store a vector for each document token resulted in large index sizes and slow retrieval. To address this, follow-up work has investigated several methods to make retrieval more tractable and reduce the total space consumed by the embeddings. Methods to reduce representation size include using clustering and residual encoding \citep{colbert_v2, plaid_late_interaction}, employing a learned per-token pruning mechanism \citep{aligner_late_interaction, colberter}, clustering document tokens to decrease redundancy \citep{clustering_multi_vector_representations}, and using a constant number of document embeddings \citep{constant_space_multi_vector_retrieval, unified_model_late_interaction}. The most relevant to our work are methods that compress the set of query and document embeddings into a single embedding such that an inner product is an approximation of MaxSim. MUVERA achieves this by using Locality Sensitive Hashing (LSH) and provides a theoretical bound on the error between the inner product of the constructed vectors and the original MaxSim similarity \citep{muvera_multi_vector_retrieval_single_vector}. LEMUR uses a data-aware approach which learns a mapping from the multi-vector representations to single-vector ones in a way that minimizes the difference between the MaxSim similarity and inner product \citep{lemur_learned_multi_vector_retrieval}.

\subsection{Theoretical Properties of Retrieval Models}
Understanding the capabilities and limitations of retrieval models is crucial to improving retrieval effectiveness and building reliable search systems. Prior work exploring these theoretical properties has largely investigated the capacity of dense retrieval models which use a single embedding (i.e. vector) to represent queries and documents, with the similarity generally computed using an inner product. \citet{sparse_dense_attentional_representations} provide a bound on the number of pair-wise errors produced by compressing a vocabulary vector into a smaller dimension. They show that, for binary vectors, the size of the embedding dimension needed to guarantee a specific error bound grows linearly with the largest number of unique terms in a document. They show that by using several vectors with smaller embedding dimension the error rate can remain the same while the embedding dimension decreases. \citet{defense_dual_encoders_ranking} show that the inner product between two embeddings can approximate an arbitrary continuous scoring function given that the dimension of the embeddings is countably infinite. \citet{lite_learned_late_interaction} show that when the output embedding dimension is smaller than the total size of the input to a transformer encoder then any approximation of a scoring function will have some irreducible error. They also show that their learned late-interaction similarity is a universal function approximation. \citet{hypencoder_hypernetworks_retrieval} shows that for any $d+1$ document vectors with dimension $d$ there is no query vector that can perfectly separate some relevance pattern using inner product. \citet{theoretical_limitations_embedding_retrieval} proved that the number of orderings that a set of document and query embeddings can express such that there is a certain margin between relevant and irrelevant documents is dependent on the embedding dimension and number of documents. \citet{on_strengths_and_limitations_of_single_vector_embeddings} observe that the results by \citet{sign_rank_versus_vc_dim} can be used to show that $2k + 1$ dimensions are necessary for any top-$k$ ranking regardless of the number of documents when using single vectors and inner products to assess similarity.
While the theoretical limits of single-vector dense retrieval have been increasingly scrutinized, the expressivity of multi-vector models remains largely unexplored. 

Concurrent work by \citet{multi_vector_embeddings_are_provably_more_expressive} studies whether MaxSim similarities produced by sets of at most $m$ unit vectors in $\mathbb{R}^d$ can be approximated by single-vector inner products in $\mathbb{R}^D$. They construct a finite family for which achieving an error $\epsilon$ requires $D=(\epsilon^2m)^{\Omega(1/\epsilon)}$, even when the single-vector representations may be chosen after observing the complete dataset and each query contains only one vector. This establishes a strong separation in one direction: arbitrary multi-vector MaxSim similarities cannot generally be compressed into comparably sized single-vector representations. This finding aligns with our own finding that single-vector finite-dimensional inner products cannot simulate arbitrary-dimensional sparse inner products. Since we show that such sparse inner products can be simulated by MaxSim, this implies that single-vector embeddings cannot simulate multi-vector embeddings of comparable dimension -- giving an analogous result to \citet{multi_vector_embeddings_are_provably_more_expressive}, but in the exact rather than approximate setting.  Unlike our work, \citet{multi_vector_embeddings_are_provably_more_expressive} does not address the reverse question: which single-vector similarities standard MaxSim can or cannot realize. %
They also do not propose extensions of MaxSim or 
demonstrate empirical retrieval gains resulting from their theoretical analysis.

\section{Late-Interaction Similarity for Exact Inner Product Computations}
\label{sec:late_interaction_for_exact_inner_product}

This section contains the proof that MaxSim similarity can exactly replicate the inner product between non-negative vectors of arbitrary dimension (possibly infinite) with representation space tied to the number of non-zero elements. We additionally show that MaxSim can be extended to enable the exact inner product replication of any real-valued vectors.

\subsection{Exact Non-Negative Inner Product Computations with MaxSim Similarity}
\label{sec:max_sim_positive_inner_product}

As our results depend on the support and sparsity of vectors, we formally define these terms.

\begin{definition}[Support and Sparsity]
The \textbf{support} of a vector $\vx$ is the set of indices corresponding to its non-zero elements, defined as $\supp(\vx) := \{i \in \N \mid x_i \neq 0\}$. A vector $\vx$ is said to be \textbf{$k$-sparse} if $|\supp(\vx)| \leq k$, i.e., the cardinality of its support is at most $k$.
\end{definition}

We now state the main result of this section: MaxSim similarity can exactly replicate the inner product between two non-negative vectors $\vu, \vv \in \R^n_{\geq0}$, where $n$ can be countably infinite, and
the size of the vector sets $\cU, \cV \subset \R^{3}$ representing the original vectors depends only on the number of non-zero elements in $\vu$ and $\vv$. Importantly, the construction of $\cU$ and $\cV$ can be done independently (i.e. $\cU$ has no knowledge of $\cV$ and vice versa), making the finding directly applicable to realistic settings such as document retrieval.

This result indicates that MaxSim is as capable as an inner product between non-negative vectors and additionally can represent an infinite-dimension non-negative sparse vector with representation space only tied to the non-zero elements. We show in Section \ref{sec:dim_bottleneck} that this is impossible for a standard vector inner product to accomplish, making MaxSim more capable in this regard.

\begin{theorem}
\label{thm:main_maxsim}
Let $\vu, \vv \in \R^n_{\geq0}$ be  non-negative vectors, where $\vu$ is $k_u$-sparse and $\vv$ is $k_v$-sparse. There exist sets of vectors $\cU, \cV \subset \R^3$, with $|\cU| = k_u$ and $|\cV| = k_v + 1$, such that:
\begin{equation}
    \langle \vu, \vv \rangle = S(\mathcal{U}, \mathcal{V}),
\end{equation}
Furthermore, the sets $\cU$ and $\cV$ can be constructed independently; the mapping of $\vu$ to $\cU$ depends strictly on $\vu$, and the mapping of $\vv$ to $\cV$ depends strictly on $\vv$, preserving the asymmetric dual-encoder architecture standard in retrieval.
\end{theorem}

We transform the sparse vectors $\vu$ and $\vv$ into sets of dense 3-dimensional vectors, denoted $\cU$ and $\cV$. The key intuition is that we can use the inner product between two vectors in $\R^3$ to represent evaluating a specific quadratic polynomial at a given point. By creating the polynomial so that when evaluated at a non-matching index the value is less than $0$ and at a matching index is a predefined value, the maximum in the similarity only allows matching indices to contribute to the final similarity.

\begin{lemma}[Polynomial Construction]
\label{lemma:poly}
For any $d \in \N$ and $w \in \R_{\ge 0}$, there exists a coefficient vector $c(d, w) =(c_0, c_1, c_2)^T \in \R^3$ defining a quadratic polynomial $p(x) = \vc^T \phi(x) = c_0 + c_1 x + c_2 x^2$ such that:
\begin{enumerate}
    \item $p(d) = w$
    \item $p(d') < 0$ for all $d' \in \N$ where $d' \neq d$.
\end{enumerate}
\end{lemma}
\begin{proof}
Consider the polynomial $p(x) = w - C(x-d)^2$ for some constant $C > 0$.
By construction, $p(d) = w - C(d-d)^2 = w$.
For any other index $d' \in \N$ where $d' \neq d$, the term $(d'-d)^2 \ge 1$.
Thus, $p(d') = w - C(d'-d)^2 \le w - C$.
To ensure $p(d') < 0$, we must choose $C$ such that $w - C < 0$, which implies $C > w$. Any such choice of $C$ satisfies the conditions. For example, let $C = w + 1$.
The polynomial can be expanded to find the coefficients:
\begin{align*}
    p(x) &= w - C(x^2 - 2dx + d^2) \\
         &= (-C)x^2 + (2Cd)x + (w - Cd^2).
\end{align*}
The corresponding coefficient vector is $c(d, w) = (w - Cd^2, 2Cd, -C)^T$. 
\end{proof}

The polynomial construction creates a function that evaluates to a given value at a certain point and negative elsewhere. We now  construct the complementary embedding that represents a specific point to evaluate. Previous work has used a similar approach to exactly factorize sparse matrices \citep{graph_exact_represenations, exact_representation_of_sparse_networks}.
\begin{definition}[Embedding Map]
Define the quadratic embedding map $\phi: \N \to \R^3$ as:
\begin{equation}
    \phi(d) = \begin{pmatrix} 1 \\ d \\ d^2 \end{pmatrix}.
\end{equation}
\end{definition}

We now prove the main theorem of this section.
\begin{proof}[Proof of Theorem \ref{thm:main_maxsim}]
We construct the sets $\cU$ and $\cV$ as follows. For the vector $\vu$, we map each non-zero entry $u_i$ to a vector in $\R^3$ using the embedding map $\phi$:
\begin{equation}
    \cU = \left\{u_i \phi(i) \mid i \in \supp(\vu) \right\}.
\end{equation}
For the vector $\vv$, we utilize the polynomial coefficients from Lemma \ref{lemma:poly}. We define:
\begin{equation}
    \cV = \left\{c(i, v_i)\mid i \in \supp(\vv) \right\} \cup \{ \vz \},
\end{equation}
where $\vz \in \R^3$ is the zero vector.

We now show that $S(\cU, \cV) = \langle \vu, \vv \rangle$. By definition of the standard inner product, $\langle \vu, \vv \rangle = \sum_{i \in \supp(\vu)} u_i v_i$. Expanding the Max-Sim similarity for our constructed sets, we have:
\begin{equation}
    S(\mathcal{U}, \mathcal{V}) = \sum_{\mathbf{m} \in \mathcal{U}} \max_{\mathbf{t} \in \mathcal{V}} \langle\mathbf{m}, \mathbf{t} \rangle = \sum_{i \in \supp(\vu)} \max_{\mathbf{t} \in \mathcal{V}} \langle u_i \phi(i), \mathbf{t} \rangle.
\end{equation}
To prove equivalence, it suffices to show that for every $i \in \supp(\vu)$, the inner term evaluates to $u_i v_i$. We analyze two cases based on the value of $v_i$:

\textbf{Case 1: $v_i = 0$ (i.e., $i \notin \supp(\vv)$).} \\
By Lemma \ref{lemma:poly}, for any $c(j, v_j) \in \cV$ (where $j \neq i$), the corresponding polynomial evaluates to a strictly negative value at $i$. Thus, $\langle u_i \phi(i), c(j, v_j) \rangle = u_i p(i) < 0$ (since $u_i > 0$). However, the zero vector $\vz \in \cV$ yields $\langle u_i \phi(i), \vz \rangle = 0$. Therefore, the maximum over $\cV$ is achieved by $\vz$, yielding $0$. This matches $u_i v_i = 0$.

\textbf{Case 2: $v_i > 0$ (i.e., $i \in \supp(\vv)$).} \\
By construction, $\cV$ contains the coefficient vector $c(i, v_i)$. By Lemma \ref{lemma:poly}, the inner product with this specific vector is $\langle u_i \phi(i), c(i, v_i) \rangle = u_i p(i) = u_i v_i$. For all other vectors $c(j, v_j) \in \cV$ ($j \neq i$), Lemma \ref{lemma:poly} guarantees $p(i) < 0$, yielding a negative inner product. The zero vector yields $0$. Since $u_i, v_i > 0$, the term $u_i v_i$ is strictly positive and is therefore the maximum over all $\mathbf{t} \in \cV$.

In both cases, $\max_{\mathbf{t} \in \mathcal{V}} \langle u_i \phi(i), \mathbf{t} \rangle = u_i v_i$. Summing over all $i \in \supp(\vu)$ yields the exact inner product, concluding the proof.
\end{proof}

\subsection{Standard MaxSim Cannot Exactly Replicate Signed Inner Products}
\label{sec:standard_maxsim_signed_limitation}

 We now ask whether the restriction that both vectors must be non-negative is necessary or whether a similarly sparsity-preserving representation can recover arbitrary signed inner products without modifying the MaxSim operation.

We allow each query or document embedding complete information on the input vector and allow the document representation to include an additional set of embeddings that are shared for all document representations similar to the $\vz$ vector in Section \ref{sec:max_sim_positive_inner_product}. The only limitation is that the number of query and document embeddings (not including those shared for all documents) must be equal to the cardinalities of the supports of the original query and document vectors. Even with these minimal constraints and a high degree of flexibility, standard MaxSim cannot recover all signed inner products in a fixed embedding dimension.

\begin{definition}[Contextual Sparsity-Preserving Encoding]
\label{def:shared_augmented_encoding}
Let $M\in\N$ be fixed. A contextual sparsity-preserving encoding consists of independently constructed mappings
\[
\cU:\R^n\rightarrow 2^{\R^M},
\qquad
\cV_{\mathrm{var}}:\R^n\rightarrow 2^{\R^M},
\]
together with a fixed finite set $\mathcal F\subset\R^M$ shared by every document representation. Here, $2^{\R^M}$ represents all possible subsets with elements from $\R^M$. The mappings satisfy
\[
|\cU(\vu)|=|\supp(\vu)|,
\qquad
|\cV_{\mathrm{var}}(\vv)|=|\supp(\vv)|,
\]
and the complete document representation is
\[
\cV(\vv)=\cV_{\mathrm{var}}(\vv)\cup\mathcal F.
\]
Each embedding in $\cU(\vu)$ and $\cV_{\mathrm{var}}(\vv)$ may depend arbitrarily on the complete corresponding input vector. The only restriction is that the encoding contains exactly one input-dependent embedding per nonzero coordinate, while the vectors in $\mathcal F$ are fixed and identical across documents. This definition fits the prior construction for non-negative vectors when $\mathcal{F} = \{\vz\}$. Additionally, in the non-negative construction the embeddings are constructed with only coordinate-wise information. 
\end{definition}

\begin{theorem}[Sparsity-Preserving Limitation of Standard MaxSim]
\label{thm:standard_maxsim_signed_limit}
Let $M,n\in\N$, let $\mathcal F$, $\cU$, and $\cV_{\mathrm{var}}$ satisfy Definition \ref{def:shared_augmented_encoding}. If
\[
S\bigl(\cU(\vu),\cV(\vv)\bigr)=\langle\vu,\vv\rangle
\]
for every pair of one-sparse vectors $\vu,\vv\in\R^n$, then $n\le M$. Consequently, when $n>M$, some pair of one-sparse vectors cannot have its inner product exactly recovered, regardless of the  encoders or shared vectors used.
\end{theorem}

\begin{proof}
For each $i,j\in[n]$, let
\[
\vu_i=-\ve_i,
\qquad
\vv_j=\ve_j,
\qquad
\langle\vu_i,\vv_j\rangle=-\delta_{ij},
\]
where $\ve_i$ denotes the $i$th standard basis vector and $\delta_{ij}$ is the Kronecker delta. Thus, query $\vu_i$ must score document $\vv_i$ as $-1$ and every document $\vv_j$ with $j\neq i$ as $0$.

Because these vectors are one-sparse, their input-dependent representations contain one embedding each. Write
\[
\cU(\vu_i)=\{\vm_i\},
\qquad
\cV_{\mathrm{var}}(\vv_j)=\{\vt_j\},
\]
for $\vm_i,\vt_j\in\R^M$. For each query, define the document-independent contribution of the shared vectors by
\[
b_i=\max_{\vf\in\mathcal F}\langle\vm_i,\vf\rangle,
\]
using the convention $b_i=-\infty$ when $\mathcal F=\emptyset$. Exact recovery therefore requires
\[
\max\!\left\{\langle\vm_i,\vt_j\rangle,b_i\right\}
=-\delta_{ij}
\qquad\text{for all }i,j\in[n].
\]

Taking $j=i$, the maximum must equal $-1$, so both of its arguments are at most $-1$. In particular,
\[
b_i\le -1
\qquad\text{and}\qquad
a_i:=\langle\vm_i,\vt_i\rangle\le -1.
\]
For $j\neq i$, the required score is $0$. Since $b_i\le-1$, the shared vectors cannot produce this score, and exact recovery forces
\[
\langle\vm_i,\vt_j\rangle=0.
\]

Hence the matrix $A\in\R^{n\times n}$ defined by
\[
A_{ij}=\langle\vm_i,\vt_j\rangle
\]
is diagonal with nonzero diagonal entries $a_i$, and therefore $\operatorname{rank}(A)=n$. On the other hand, if $U,V\in\R^{n\times M}$ have rows $\vm_i^T$ and $\vt_i^T$, respectively, then $A=UV^T$, implying
\[
\operatorname{rank}(A)\le M.
\]
Thus $n\le M$, proving the result.
\end{proof}

The proof shows that MaxSim is unable to exactly replicate even a simple inner product reconstruction between real-valued vectors. The shared document vectors $\mathcal F$ are of limited usefulness, since for a fixed query $\vu_i$, their entire contribution is limited to the same value $b_i$ for every document, forcing the input-dependent  document embeddings $\vt_j$ to realize a full-rank diagonal similarity matrix.

The zero vector used in Section \ref{sec:max_sim_positive_inner_product} is an especially direct example. If $\vz\in\mathcal F$, then every query has similarity $0$ with a shared document vector, making every MaxSim score non-negative. Standard MaxSim then cannot recover even the one-dimensional inner product $\langle-\ve_i,\ve_i\rangle=-1$.

Theorem \ref{thm:standard_maxsim_signed_limit} does not mean that standard MaxSim can never produce negative scores. Rather, it shows that standard MaxSim cannot extend the earlier sparsity-preserving construction to arbitrary real-valued vectors while retaining a fixed embedding dimension. This motivates our extension, detailed in the following section, which overcomes this limitation with standard MaxSim by decomposing the original vector values into two parts, magnitude and sign, which remain disentangled until after the maximum step in the similarity.

\subsection{Extension to MaxSim Similarity for Exact Real Number Inner Product Computations}
In Section \ref{sec:max_sim_positive_inner_product}, we showed that MaxSim similarity can exactly replicate the inner product between two non-negative vectors with representation size corresponding to the number of non-zero indices. In Section \ref{sec:standard_maxsim_signed_limitation}, we proved that with a sparsity-preserving encoder standard MaxSim cannot replicate exact inner products between real-valued vectors. In this section, we address this limitation by extending MaxSim and multi-vector representations to enable the exact inner product to be replicated for any real-valued vectors $\vu, \vv \in \R^n$, where, as before, $n$ can be countably infinite. The key modification is to handle negative values by decoupling each entry into its magnitude and sign. The embedding method used in Section \ref{sec:max_sim_positive_inner_product} is applied strictly to the magnitudes, so the maximum in the similarity still properly selects matching indices, while the scalar signs are incorporated after the maximum to reconstruct the correct final product.

\subsubsection{Modified Vector Representation and Similarity}

Any non-zero real number $x$ can be decomposed as $x = \sgn(x) \cdot |x|$. We adapt the sets $\cU$ and $\cV$ from Theorem \ref{thm:main_maxsim} to store these signs alongside the embedded vectors. We start by modifying MaxSim to allow the incorporation of sign values after the maximum operation.

\begin{definition}[Signed MaxSim Similarity]
Let $\cU_s$ and $\cV_s$ be sets of pairs, where each pair consists of a vector and a scalar sign. 
For each query pair $(\vm, s_q) \in \cU_s$, let $(\vt^*(\vm), s_d^*(\vm))$ denote the specific pair in $\cV_s$ that maximizes the inner product with $\vm$:
\begin{equation}
    (\vt^*(\vm), s_d^*(\vm)) = \argmax_{(\vt, s_d) \in \cV_s} \langle \vm, \vt \rangle.
\end{equation}
The Signed MaxSim similarity $S_{\pm}$ is defined as the sum over all query pairs of their maximum inner product, multiplied by the signs of both the query vector and the maximizing document vector:
\begin{equation}
    S_{\pm}(\cU_s, \cV_s) = \sum_{(\vm, s_q) \in \cU_s} s_q \cdot s_d^*(\vm) \cdot \langle \vm, \vt^*(\vm) \rangle.
\end{equation}
\end{definition}

To utilize this new similarity function, we must also adapt our vector representations to store these scalar signs alongside the embedded magnitudes.
\begin{definition}[Signed Vector Sets]
Let $\vu, \vv \in \R^n$ be real-valued vectors. We represent them as sets of pairs, $\cU_s, \cV_s \subset \R^3 \times \{-1, 1\}$, where each pair consists of a vector and a scalar sign.

For the vector $\vu$, we map each non-zero entry $u_i$ using its magnitude $|u_i|$ and sign $\sgn(u_i)$:
\begin{equation}
    \cU_s = \left\{ \big(|u_i| \phi(i), \sgn(u_i)\big) \mid i \in \supp(\vu) \right\}.
\end{equation}

For the vector $\vv$, let $c(i, |v_i|)$ be the polynomial coefficient vector generated by Lemma \ref{lemma:poly}. We define:
\begin{equation}
    \cV_s = \left\{ \big(c(i, |v_i|), \sgn(v_i)\big) \mid i \in \supp(\vv) \right\} \cup \{ (\vz, 1) \},
\end{equation}
where $(\vz, 1)$ is the zero pair, included to serve as a baseline for non-matching indices.
\end{definition}

\subsubsection{Inner Product Recovery}
With the signed vector representations and the modified similarity function defined, we now prove that this formulation exactly recovers the inner product for any arbitrary real-valued vectors.

\begin{theorem}
Let $\vu, \vv \in \R^n$ be arbitrary real-valued vectors, where $\vu$ is $k_u$-sparse and $\vv$ is $k_v$-sparse. There exist sets of pairs $\cU_s, \cV_s \subset \R^3 \times \{-1, 1\}$, with $|\cU_s| = k_u$ and $|\cV_s| = k_v + 1$, such that:
\begin{equation}
    \langle \vu, \vv \rangle = S_{\pm}(\cU_s, \cV_s).
\end{equation}
Importantly, like in Theorem \ref{thm:main_maxsim}, $\cU_s$ and $\cV_s$ can be produced with only access to $\vu$ and $\vv$, respectively, with no knowledge of the other vector until the similarity is calculated.
\end{theorem}

We extend the construction from Theorem \ref{thm:main_maxsim} by operating exclusively on the absolute values of the vector entries. Because magnitudes are strictly positive, the inner maximization step correctly identifies matching indices exactly as it did in the positive-only case. Once the matching index is found (or the zero vector is selected in the case of a mismatch), the original signs are multiplied back into the resulting magnitude to recover the exact real-valued product.

\begin{proof}
We construct the sets $\cU_s$ and $\cV_s$ as defined above. We analyze the term contributed to the Signed Max-Sim sum by each pair $(\vm_i, s_i) \in \cU_s$, where $\vm_i = |u_i| \phi(i)$ and $s_i = \sgn(u_i)$ for some $i \in \supp(\vu)$. 

The maximization step, $\argmax_{(\vt, s_d) \in \cV_s} \langle \vm_i, \vt \rangle$, depends only on the vector components and operates on strictly positive magnitudes ($u_i > 0$ and $v_i > 0$). Therefore, the logic follows the proof of Theorem \ref{thm:main_maxsim}. We analyze two cases based on whether the index $i$ is present in $\vv$:

\textbf{Case 1: $i \notin \supp(\vv)$ (Mismatched Index).} \\
By Lemma \ref{lemma:poly}, for all $(c(j, |v_j|), \sgn(v_j)) \in \cV_s$, the inner product $\langle |u_i| \phi(i), c(j, |v_j|) \rangle = |u_i| p(i) < 0$. However, the zero pair yields $\langle |u_i| \phi(i), \vz \rangle = 0$. Thus, the unique maximum is $0$, achieved by the pair $(\vt^*(\vm_i), s_d^*(\vm_i)) = (\vz, 1)$. The contribution to the sum is:
\begin{equation}
    s_i \cdot s_d^*(\vm_i) \cdot \langle \vm_i, \vt^*(\vm_i) \rangle = \sgn(u_i) \cdot 1 \cdot 0 = 0.
\end{equation}
This correctly matches the product $u_i v_i = 0$.

\textbf{Case 2: $i \in \supp(\vv)$ (Matching Index).} \\
By construction, $\cV_s$ contains the pair $(c(i, |v_i|), \sgn(v_i))$. By Lemma \ref{lemma:poly}, the inner product with this specific vector is $\langle |u_i| \phi(i), c(i, |v_i|) \rangle = |u_i| |v_i| > 0$. For all other vectors $c(j, |v_j|) \in \cV_s$ ($j \neq i$), the inner product is strictly negative, and the zero vector yields $0$. Thus, the unique maximum is $|u_i| |v_i|$, achieved by the pair $(\vt^*(\vm_i), s_d^*(\vm_i)) = (c(i, |v_i|), \sgn(v_i))$. The contribution to the sum is:
\begin{align*}
    s_i \cdot s_d^*(\vm_i) \cdot \langle \vm_i, \vt^*(\vm_i) \rangle &= \sgn(u_i) \cdot \sgn(v_i) \cdot \big( |u_i| |v_i| \big) \\
    &= \big( \sgn(u_i) |u_i| \big) \cdot \big( \sgn(v_i) |v_i| \big) \\
    &= u_i v_i.
\end{align*}

In both cases, the $\argmax$ is uniquely defined, and the term contributed by the index $i \in \supp(\vu)$ is exactly $u_i v_i$. Summing over all pairs in $\cU_s$ yields:
\begin{equation}
    S_{\pm}(\cU_s, \cV_s) = \sum_{i \in \supp(\vu)} u_i v_i = \langle \vu, \vv \rangle.
\end{equation}
This completes the proof.
\end{proof}

\section{The Dimensionality Bottleneck of Standard Inner Products}
\label{sec:dim_bottleneck}
In Theorem \ref{thm:main_maxsim}, we demonstrated that MaxSim similarity can exactly reproduce the inner product of two non-negative $k$-sparse vectors of arbitrary dimension using finite sets of low-dimensional vectors. Given this capability, it is natural to ask whether single-vector representations can achieve a similar compression. In this section, we demonstrate a fundamental rank constraint: no finite-dimensional inner product space can exactly preserve the inner products of arbitrarily high-dimensional sparse vectors, establishing a strict theoretical separation between single-vector and multi-vector retrieval paradigms.

To prove that embedding high-dimensional sparse vectors into a lower finite dimension $d$ is impossible, it suffices to show that we cannot even embed $d+1$ mutually orthogonal vectors into $\R^d$. We establish this lower bound below. To ensure our impossibility result is as general as possible, we allow for asymmetric embeddings, where the query and document vectors can be processed by entirely different mapping functions.

\begin{theorem}[Impossibility of Dimensionality Compression]
\label{thm:compression_impossibility}
Let $\mathcal{S} = \{\ve_1, \dots, \ve_{d+1}\} \subset \R^{d+1}$ be the set of $d+1$ standard basis vectors. For any finite dimension $d \in \N$, there are no mappings $f, g: \mathcal{S} \to \R^d$ that exactly preserve pairwise inner products. That is, there are no $f$ and $g$ such that for all $\vx, \vy \in \mathcal{S}$:
\begin{equation}
    \langle f(\vx), g(\vy) \rangle = \langle \vx, \vy \rangle.
\end{equation}
\end{theorem}

\begin{proof}
Assume such mappings $f$ and $g$ exist. Let $U, V \in \R^{d \times (d+1)}$ be the matrices whose $i$-th columns are the mapped vectors $f(\ve_i)$ and $g(\ve_i)$, respectively. The matrix of pairwise inner products between the mapped vectors is given by the product $U^T V \in \R^{(d+1) \times (d+1)}$. 

By our assumption, $U^T V$ must exactly equal the Gram matrix of the original vectors in $\mathcal{S}$, which is the identity matrix $I_{d+1}$. However, the rank of a matrix product is bounded by its inner dimension, meaning $\text{rank}(U^T V) \le d$. This implies $d+1 = \text{rank}(I_{d+1}) \le d$, which is a contradiction.
\end{proof}

This rank bottleneck dictates that standard inner products cannot exactly compress even $d+1$ orthogonal concepts into $d$ dimensions. Consequently, representing vectors from an arbitrarily large vocabulary in a fixed-dimensional space inherently requires approximation. 

This highlights a core advantage of the MaxSim architecture. While a standard inner product is strictly constrained by the linear algebraic rank of its embedding dimension, MaxSim represents a $k$-sparse vector as a set of $k$ vectors in $\R^3$. By requiring only $O(k)$ parameters, MaxSim bypasses this rank constraint entirely, perfectly preserving exact inner products regardless of the underlying vocabulary size.

\section{Late-Interaction Similarity as Logical Expression Evaluation}
\label{sec:logical_evaluation}

In this section, we explore the logical expressivity of MaxSim. We demonstrate how it naturally acts as an aggregation of fuzzy logical operations and prove its capability to evaluate positive Conjunctive Normal Form expressions (i.e. Conjunctive Normal Form expressions without negations).

\subsection{Motivation: MaxSim as a Fuzzy Logical OR}
In the previous section, we demonstrated that MaxSim can replicate the exact inner product of two non-negative vectors by mapping each non-zero entry to a 3-dimensional embedding. While this establishes MaxSim's capacity to represent standard vector spaces, it does not fully explain the benefit of MaxSim. 

Standard inner products compute a sum over all matching features. In information retrieval, this inherently acts as a "fuzzy AND" or an accumulation of evidence: a document with both "red" and "blue" features will score higher than a document with only "red", even if the user only wanted one or the other. 

MaxSim, however, computes a sum of \textit{maximums}. If we can encode a set of alternative concepts (e.g., "red", "crimson", "scarlet") into a \textit{single} query vector, the inner maximization of MaxSim will score the document based solely on its single best matching feature. It will not over-reward documents that contain multiple synonyms. Similarly, an OR relationship between concepts like "red" and "blue" can be encoded into a single query vector, ensuring that documents containing both terms are not disproportionately rewarded when the user only desires one. In this way, MaxSim naturally executes a \textbf{fuzzy logical OR} over grouped query terms. 

In this section, we formalize this intuition. We show that by increasing the embedding dimension proportionally to the size of the clause, a single query vector can encode an arbitrarily large OR-clause of terms. We then prove that this mechanism allows standard MaxSim to evaluate positive Conjunctive Normal Form (CNF) Boolean logic so that the ordering of documents matches the exact evaluation of the expression.

\subsection{Formalism for Grouped Queries (Fuzzy OR)}

To formalize the intuition of MaxSim as a fuzzy OR, we first define the mathematical structures representing documents and grouped queries. Let our universe of keys be the natural numbers $\N$ and our values be non-negative real numbers in $\R_{\ge 0}$. 

\begin{definition}[Document and Query Groups]~\\
\label{def:document_and_query_group_sets}
Let $D = \{ (d_1, w_1), \dots, (d_{k_d}, w_{k_d}) \}$ be a \textbf{document set} of key-value pairs, where $d_i \in \N$ are unique keys and $w_i \in \R_{\ge 0}$ are weights. 

Let $S = \{ (d_{q,1}, w_{q,1}), \dots, (d_{q,m}, w_{q,m}) \}$ be a \textbf{query group}. The group $S$ is a set of key-value pairs representing a single OR-clause of terms, where $d_{q,i} \in \N$ are unique keys and $w_{q,i} \in \R_{\ge 0}$ are weights. Let $D_S \subset \N$ denote the set of keys present in query group $S$.
\end{definition}

In standard fuzzy logic, the logical OR of a set of truth values $x_1, \dots, x_n \in[0, 1]$ is defined by the maximum operator: $\bigvee_{i=1}^n x_i = \max(x_1, \dots, x_n)$. To apply this to information retrieval, we must extend this concept in two natural ways.

First, rather than strict $[0, 1]$ truth values, a document's match to a specific term is typically represented by an unbounded, non-negative relevance weight $w_i \in \R_{\ge 0}$. It is natural that a document might be highly relevant to one term but only marginally relevant to another. 

Second, not all terms within an OR-clause are equally valuable to the user. For example, an exact keyword match might be highly desired, while a broader synonym is acceptable but less preferred. We can represent this by assigning a query weight $w_{q,i} \in \R_{\ge 0}$ to each term $d_{q,i}$ in the query group $S$. 

When a document contains a matching term, its overall score for that specific term is naturally the product of the document's relevance weight and the query's importance weight: $w_i w_{q,i}$. To evaluate the entire OR-clause, we apply the standard fuzzy logic principle of taking the maximum over all available options. This yields a weighted, unbounded generalization of the logical OR, which we refer to as a \textbf{Weighted Max-OR}.

\begin{definition}[Weighted Max-OR] 
Given a document set $D$ and a query group $S$ representing an OR-clause of terms, we define the \textbf{Weighted Max-OR} evaluation of $S$ on $D$ as the maximum combined weight across all matching terms:
\begin{equation}
    \text{MaxOR}(S, D) = \max \left( \{0\} \cup \{ w_i w_{q,i} \mid (d_i, w_i) \in D \text{ and } (d_i, w_{q,i}) \in S \} \right).
\end{equation}
The set $\{0\}$ is included so that when there are no shared keys between $D$ and $S$ the maximum is still defined and evaluates to $0$ like in standard fuzzy logic.
\end{definition}

This definition provides a natural extension of Boolean logic to continuous retrieval scores. By taking the maximum over the matching terms in a query group, we select the single strongest weighted match for that concept. This perfectly captures the intent of an OR-clause: a document is rewarded for its best-matching synonym, but its score is not artificially inflated if it happens to contain multiple synonyms (which would happen if we used a sum, representing an accumulation of matches or a fuzzy AND).

With the Weighted Max-OR defined for a single clause, we now show that MaxSim can evaluate an entire set of these clauses simultaneously.

\begin{theorem}[Weighted Max-OR Evaluation]
\label{thm:grouped_queries}
Let $D$ be a document set and $Q = \{ S_1, \dots, S_{k_q} \}$ be a grouped query set consisting of multiple query groups. Let $M = 2 \max_{S_j \in Q} |S_j| + 1$. There exist sets of vectors $\cQ, \cD \subset \R^M$, with $|\cQ| = |Q|$ and $|\cD| = |D| + 1$, such that the MaxSim similarity computes the sum of the Weighted Max-OR evaluations for each query group:
\begin{equation}
    S(\cQ, \cD) = \sum_{S_j \in Q} \text{MaxOR}(S_j, D).
\end{equation}
Importantly, the sets $\cQ$ and $\cD$ can be constructed independently; the mapping function applied to $Q$ requires no knowledge of $D$, and the mapping applied to $D$ requires no knowledge of $Q$.
\end{theorem}

To prove this theorem, we first generalize our embedding map to arbitrary dimensions, which will allow us to construct polynomials that encode entire OR-clauses.

\begin{definition}[Generalized Embedding Map]
\label{def:generalized_embedding_map}
For a given dimension $M \in \N$ where $M \neq 0$, define the generalized embedding map $\phi_M: \N \to \R^M$ as:
\begin{equation}
    \phi_M(d) = \begin{pmatrix} 1 & d & d^2 & \dots & d^{M-1} \end{pmatrix}^T.
\end{equation}
\end{definition}

To encode an entire group $S$ into a single vector, we construct a polynomial that interpolates the weights of the desired keys while evaluating to a negative number for all irrelevant keys.

\begin{lemma}[OR-Clause Polynomial Construction]
\label{lemma:multi_key_poly}
For any query group $S$ of size $m = |S|$, there exists a coefficient vector $c(S)  = \vc_S \in \R^{2m+1}$ defining a polynomial $p_S(t) = \vc_S^T \phi_{2m+1}(t)$ such that:
\begin{enumerate}
    \item $p_S(d) = w$ for all $(d, w) \in S$.
    \item $p_S(d') < 0$ for all keys $d' \in \N \setminus D_S$.
\end{enumerate}
\end{lemma}

\begin{proof}
Let $x_S(t)$ be the Lagrange interpolating polynomial of degree at most $m - 1$ that passes through all points in $S$. By definition, $x_S(d) = w$ for all $(d, w) \in S$.

Define the root polynomial $y_S(t) = \prod_{d \in D_S} (t - d)$. This polynomial has degree $m$ and evaluates to zero exactly on the keys in $D_S$.

We define our target polynomial as $p_S(t) = x_S(t) - C \cdot (y_S(t))^2$ for some value $C > 0$ selected based on the pairs in $S$. 
For any target key $d \in D_S$, $y_S(d) = 0$, so $p_S(d) = x_S(d) = w$, satisfying the first condition.

For any non-target key $d' \in \N \setminus D_S$, we require $p_S(d') < 0$, which is equivalent to requiring $C > \frac{x_S(d')}{(y_S(d'))^2}$. Because $d'$ and all $d \in D_S$ are distinct integers, the difference $(d' - d)$ is a non-zero integer, meaning $(y_S(d'))^2 \ge 1$. This ensures the denominator is never zero and does not shrink arbitrarily small. Furthermore, the degree of $(y_S(t))^2$ is $2m$, which is strictly greater than the degree of $x_S(t)$ (which is at most $m - 1$). Therefore, as $t \to \infty$, the ratio $\frac{x_S(t)}{(y_S(t))^2} \to 0$. Because the ratio goes to zero asymptotically and is evaluated over discrete integers, it is bounded above by some finite maximum value $B$ for all $d' \in \N \setminus D_S$.

By choosing $C > \max(0, B)$, we guarantee $p_S(d') < 0$ for all $d' \notin D_S$. The maximum degree of $p_S(t)$ is the degree of $(y_S(t))^2$, which is $2m$. Thus, it has $2m+1$ coefficients, which can be represented by the vector $\vc_S \in \R^{2m+1}$.
\end{proof}

Using this polynomial construction, we can now show that the inner maximization step of MaxSim perfectly computes the Weighted Max-OR for a single query group.

\begin{lemma}[Inner Maximization for a Single Query Group]
\label{lemma:inner_max}
Let $S$ be a query group and $D$ be a document set. Let $M \ge 2|S| + 1$. Let $c(S) = \vc_S \in \R^M$ be the polynomial coefficient vector from Lemma \ref{lemma:multi_key_poly} (padded with zeros if $2|S| + 1 < M$). Let $\cD = \{ w_i \phi_M(d_i) \mid (d_i, w_i) \in D \} \cup \{ \vz \}$ be the document embeddings. Then the maximum inner product evaluates to the Weighted Max-OR:
\begin{equation}
    \max_{\vd \in \cD} \langle \vc_S, \vd \rangle = \text{MaxOR}(S, D).
\end{equation}
\end{lemma}

\begin{proof}
For any document vector $\vd_i = w_i \phi_M(d_i) \in \cD$, the inner product evaluates the polynomial: $\langle \vc_S, w_i \phi_M(d_i) \rangle = w_i p_S(d_i)$. 

If $d_i \notin D_S$ (a non-matching key), Lemma \ref{lemma:multi_key_poly} guarantees $p_S(d_i) < 0$. Since document weights $w_i \ge 0$, the inner product is $\le 0$. The zero vector $\vz \in \cD$ yields $\langle \vc_S, \vz \rangle = 0$, ensuring these non-matches result in a score of at most $0$.

If $d_i \in D_S$ (a matching key), let $(d_i, w_{q,i})$ be the corresponding pair in $S$. Lemma \ref{lemma:multi_key_poly} guarantees $p_S(d_i) = w_{q,i}$. The inner product is exactly $w_i w_{q,i} \ge 0$ as both $w_i \ge 0$ and $w_{q, i} \ge 0$.

Therefore, the set of all inner products evaluated by the $\max$ operation is exactly $\{0\} \cup \{ w_i w_{q,i} \mid (d_i, w_i) \in D \text{ and } d_i \in D_S \}$. The maximum of this set is, by definition, $\text{MaxOR}(S, D)$.
\end{proof}

We now prove the main theorem of this section.
\begin{proof}[Proof of Theorem \ref{thm:grouped_queries}]
We construct $\cD$ by mapping each document pair $(d_i, w_i) \in D$ using the generalized embedding map $\phi_M$ (Definition \ref{def:generalized_embedding_map}):
\begin{equation}
    \cD = \left\{ w_i \phi_M(d_i) \mid (d_i, w_i) \in D \right\} \cup \{ \vz \}.
\end{equation}

We construct $\cQ$ by mapping each query group $S_j \in Q$ to its polynomial coefficient vector $\vc_{S_j} = c(S_j)$ from Lemma \ref{lemma:multi_key_poly} (padded with zeros to length $M$):
\begin{equation}
    \cQ = \left\{ \vc_{S_j} \mid S_j \in Q \right\}.
\end{equation}

Expanding the MaxSim similarity $S(\cQ, \cD)$, we have:
\begin{equation}
    S(\cQ, \cD) = \sum_{\vc_{S_j} \in \cQ} \max_{\vd \in \cD} \langle \vc_{S_j}, \vd \rangle.
\end{equation}

By Lemma \ref{lemma:inner_max}, the inner maximization for each query group $S_j \in Q$ exactly computes $\text{MaxOR}(S_j, D)$. Substituting this into the sum yields:
\begin{equation}
    S(\cQ, \cD) = \sum_{S_j \in Q} \text{MaxOR}(S_j, D).
\end{equation}
This completes the proof.
\end{proof}

\subsection{Exact Boolean Logic (Positive CNF)}

Having established that MaxSim evaluates a Weighted Max-OR, we now show that if we restrict the weights to binary values, MaxSim exactly evaluates standard Conjunctive Normal Form (CNF) Boolean logic without negations, so that a set of scored documents is rank equivalent to the exact evaluation. We start by formally defining rank equivalence.

\begin{definition}[Rank Equivalence]
Let $X$ be a set of items. Given a target scoring function $s: X \to \R$ and an evaluated scoring function $s': X \to \R$, we say $s'$ is \textbf{rank equivalent} to $s$ if for all $x, y \in X$:
\begin{equation}
    s(x) > s(y) \implies s'(x) > s'(y).
\end{equation}
This is useful for asserting equivalence for ranking tasks where the relative ordering matters, but the exact scores do not. 
\end{definition}

Next, we define the structure of a Boolean document and a positive CNF query.
\begin{definition}[Boolean Document and Positive CNF Query]
Let a document $D \subset \N$ be a set of unique keys. Let $K = T_1 \wedge T_2 \wedge \dots \wedge T_h$ be a logical query in Conjunctive Normal Form (CNF), where each clause $T_j = t_{j,1} \vee t_{j,2} \vee \dots \vee t_{j,l_j}$ is a disjunction (OR) of positive keys. 

The strict Boolean scoring function $s_K: 2^\N \to \{0, 1\}$ is defined as $s_K(D) = 1$ if $D$ satisfies $K$ (i.e., $D$ contains at least one key from every clause $T_j$), and $s_K(D) = 0$ otherwise.
\end{definition}

We now present the main result of this subsection: MaxSim's ability to evaluate positive CNF queries.
\begin{theorem}[MaxSim Evaluates Positive CNF]
\label{thm:soft_logic}
Let $K$ be a positive CNF query with $h$ clauses. There exists a mapping from $K$ to a query set $\cQ \subset \R^{M}$, and a mapping from any document $D$ to a document set $\cD \subset \R^{M}$, such that the MaxSim similarity $S(\cQ, \cD)$ is rank equivalent to the strict Boolean scoring function $s_K(D)$, where $|\cQ| = h$ and $|\cD| = |D| + 1$. Additionally, the sets $\cQ$ and $\cD$ can be constructed independently; the mapping function applied to $K$ requires no knowledge of $D$, and the mapping applied to $D$ requires no knowledge of $K$.
\end{theorem}

\begin{proof}
We construct $\cD$ by assigning a unit weight to each key in the document: $D_{set} = \{(d_i, 1) \mid d_i \in D\}$. We construct $\cQ$ by creating a query group for each clause $T_j$, assigning unit weights to its keys: $S_j = \{(t, 1) \mid t \in T_j\}$. We then generate $\cD$ and $\cQ$ in $\R^M$ as defined in Theorem \ref{thm:grouped_queries}.

By Theorem \ref{thm:grouped_queries}, the MaxSim similarity is:
\begin{align}
    S(\cQ, \cD) &= \sum_{S_j \in Q} \text{MaxOR}(S_j, D_{set}) \\
    &= \sum_{S_j \in Q} \max \left( \{0\} \cup \{ w_i w_{j,i} \mid (d_i, w_i) \in D_{set} \text{ and } (d_i, w_{j,i}) \in S_j \} \right).
\end{align}

Because all weights are exactly $1$, the product $w_i w_{j,i} = 1$ for all matching keys. The condition that a key exists in both $D_{set}$ and $S_j$ is equivalent to stating $D \cap T_j \neq \emptyset$.  

Thus, for a single clause $T_j$, the Weighted Max-OR evaluates to $1$ if $D \cap T_j \neq \emptyset$ (the clause is satisfied), and $0$ if $D \cap T_j = \emptyset$ (the clause is unsatisfied). The total similarity $S(\cQ, \cD)$ is exactly the number of clauses in $K$ satisfied by $D$.

To prove rank equivalence, let $D_a$ and $D_b$ be two documents such that $s_K(D_a) > s_K(D_b)$. Because $s_K$ outputs values in $\{0, 1\}$, this implies $s_K(D_a) = 1$ (satisfies all $h$ clauses) and $s_K(D_b) = 0$ (fails at least one clause). 
Consequently, $S(\cQ, \cD_a) = h$, and $S(\cQ, \cD_b) \le h - 1$. Therefore, $S(\cQ, \cD_a) > S(\cQ, \cD_b)$, satisfying the definition of rank equivalence.
\end{proof}

\section{Experiments}
\label{sec:experiments}
In Section \ref{sec:late_interaction_for_exact_inner_product}, we introduce a new late-interaction similarity function, Signed MaxSim $S_{\pm}$, which has theoretical capacities beyond standard MaxSim. In this section, we investigate whether these capacities translate to empirical gains. To better quantify the differences between these methods, we focus on a retrieval task that is more likely to benefit from the ability to replicate any real-valued inner product. Specifically, we focus on retrieval tasks that feature negations or exclusions. We refer to the model trained with the Signed MaxSim similarity $S_{\pm}$ as \textbf{Fallon} throughout our experiments.

\subsection{Datasets}
\label{sec:datasets}

To isolate the architectural impact on retrieval performance and enable better introspection, we use synthetic retrieval tasks for both training and evaluation. Inspired by the LIMIT dataset \citep{theoretical_limitations_embedding_retrieval}, documents are in the format \textit{John Smith likes Alternative Rock, Angel Food Cakes, Apricots, Argentine Ants, Avocados, Blackberries, Cabbages, and Cantaloupes} while queries ask for a person with a specific set of attributes \textit{People who like Alternative Rock and Blackberries but do not like Cherries}. Unlike the original LIMIT dataset, which only has single-attribute queries, our queries can include multiple features required for relevance and additionally can have attributes which should not be present, e.g. \textit{but do not like Cherries}. By using synthetic data we can remove possible confounding factors like validation-test mismatch and limited training data, while evaluating the abilities we believe are important differentiators between standard MaxSim and Signed MaxSim. Additionally, it makes it easy to construct different variants to test robustness and simulate out-of-domain generalization.

\subsubsection{Training Data}

For training, we generate 100k queries with 200k documents using the original LIMIT attributes and method for name generation. During training, we used contrastive training with in-batch negatives, so we constructed queries to minimize the total number of positive documents to limit the risk of in-batch negatives being true positives. Each query can have one to four inclusion terms (e.g. non-negated terms) and all queries have one negated term which should not be in relevant documents. Each training query is bundled with 32 hard negatives: documents that are relevant except for the fact that they include the negated term. Additional random negatives were included if there were fewer than 32 hard negatives.

\subsubsection{Evaluation Benchmarks}

For evaluation we investigate three different datasets:
\begin{description}[leftmargin=8pt, labelindent=0pt]
    \item[In-Domain] This dataset uses the exact same setup as the training data in terms of query formatting and vocabulary. All queries in this dataset are distinct from those seen during training, but the documents may be the same, although the documents are generated separately. 
    \item[Different Vocabulary] This dataset uses the same setup as the in-domain data, except the vocabulary is different. We generate additional vocabulary terms using Gemini conditioned on the original vocabulary. This dataset provides insight into the generalization of the model to new features. 
    \item[Negation Only] This dataset uses the same vocabulary as the in-domain data, but changes the structure of queries to contain only negations; for example, \textit{People who do not like Cherries}. We also allow one or two negations, such as \textit{People who do not like Cherries or Watermelon}, 50\% of queries have a single negation. This dataset tests the models' generalization to new query formats.
\end{description}

Each benchmark contains 2,000 queries evaluated against a corpus of 100k documents. Evaluation documents are generated independently from training documents, though incidental overlap is possible.

\subsection{Training Setup}
\label{sec:training_setup}

We train both models using a contrastive loss with in-batch negatives and a ModernBERT backbone. The main differences are the scoring function (MaxSim vs. Signed MaxSim) and the additional weight (i.e. sign) generation mechanism for Fallon. Below, we describe the loss function, model architecture, and optimization details.

All models are trained with a contrastive cross-entropy loss using in-batch negatives combined with hard negatives. Let $\mathbf{e}^q$ denote the query representation, $\mathbf{e}^{d^+}$ the positive document representation, and $\mathcal{N}$ the set of negatives. The loss uses a jointly learned temperature $\tau$:
$$
\mathcal{L}_{\text{contrastive}} = - \log \frac{\exp(S(\mathbf{e}^q, \mathbf{e}^{d^+}) / \tau)}{\exp(S(\mathbf{e}^q, \mathbf{e}^{d^+}) / \tau) + \sum_{d^- \in \mathcal{N}} \exp(S(\mathbf{e}^q, \mathbf{e}^{d^-}) / \tau)}.
$$

Both models use ModernBERT \citep{modernbert} as the backbone, which has a hidden dimension of $d = 768$ and 149\,M parameters. We prepend every query with the prefix \texttt{"query:~"} and truncate both queries and passages to a maximum of 196 tokens. No query-expansion tokens are added following contemporary late-interaction designs \citep{colbert_zero}. Before scoring, each token's hidden state is projected from $d = 768$ down to $m = 128$ by a five-layer MLP block. For the real-valued model, we use another five-layer MLP to produce the final per-token weights. The embeddings of both models are $\ell_2$-normalized after projection following \citep{colbert, colbert_v2}. We tried using non-normalized embeddings for Fallon, but found it performed worse. We believe this is due to the unbounded document scores resulting in degraded training signal when using contrastive loss. We encourage future work to investigate alternative methods to normalize training.

The baseline ColBERT model otherwise follows the design of \citet{colbert_v2} and encodes queries and documents with a single shared encoder. Relevance is measured by $\mathrm{MaxSim}$ (Definition \ref{def:maxsim}), and the contrastive loss temperature~$\tau$ is learned jointly with all other parameters, starting from an initial value of $\tau_0 = 1$.

To implement the Fallon model (i.e. our model trained with Signed MaxSim), there are minimal changes needed when compared with a standard ColBERT model; the main change is the additional MLP used for weight generation. In the theoretical version, we use the sign to produce the weight, which would limit the weights to $1$ or $-1$, presenting a problem for gradient descent. To mitigate this, we allow the weights to be any real value. This allows for gradients to flow without limitations. We tried using Tanh activation as well, but found that with $\ell_2$-normalized embeddings training was unstable, thus we stayed with no activation.

The MLP blocks map a token hidden state $\mathbf{h} \in \mathbb{R}^{768}$ to an output $\mathbf{e} \in \mathbb{R}^{d_{\mathrm{out}}}$, where $d_{\mathrm{out}}$ is either 128 (for embedding projections) or 1 (for the weight generator). The complete implementation is shown below:
\begin{align*}
    \mathbf{z}_1 &=
        \mathrm{LN}\!\bigl(
            \mathrm{GELU}(W_1\,\mathbf{h} + \mathbf{b}_1)
        \bigr),
    \\[4pt]
    \mathbf{z}_l &=
        \mathbf{z}_{l-1} \;+\;
        \mathrm{LN}\!\bigl(
            \mathrm{GELU}(W_l\,\mathbf{z}_{l-1} + \mathbf{b}_l)
        \bigr),
    \qquad l = 2,3,4,
    \\[4pt]
    \mathbf{e} &= W_5\,\mathbf{z}_4 + \mathbf{b}_5.
\end{align*}
Here $W_1,\ldots,W_4 \in \mathbb{R}^{768 \times 768}$ and $W_5 \in \mathbb{R}^{d_{\mathrm{out}} \times 768}$, with $\mathrm{LN}$ denoting layer normalization and $\mathrm{GELU}$ representing the GELU activation function \citep{gelu}.

All models are implemented using Transformers \citep{huggingface_transformers} and PyTorch \citep{pytorch}. We use the AdamW optimizer \citep{adamw} with default hyperparameters and no weight decay. All experiments use BF16 mixed-precision training with TF32 tensor operations. The learning rates were selected based on prior experiments and an LR sweep for ColBERT. We swept learning rates over $\{5 \times 10^{-6},\, 1 \times 10^{-5},\, 2 \times 10^{-5}\}$ and selected $1 \times 10^{-5}$ based on validation nDCG@10. The complete training hyperparameter configurations are available in Table \ref{tab:training_config}. During training, we evaluate each model on a validation dataset every 2,500 steps and select the final model based on validation nDCG@10. We stopped training after the validation metric showed no improvement for three consecutive checkpoints. The validation dataset uses the same setup as the In-Domain evaluation benchmark but with 1,000 queries and 40k documents. We ensure the queries are disjoint from both training and evaluation data.

\begin{center}
\small
\captionof{table}{Training hyperparameters per model. \emph{Eff.\ Batch} denotes effective batch size (per-device batch $\times$ gradient accumulation steps). \emph{Hard Neg.}\ is the number of hard negatives per positive sampled in a training batch. \emph{Attn.\ Drop} and \emph{MLP Drop} are the dropout ratios set for the backbone ModernBERT model.}
\label{tab:training_config}
\resizebox{\textwidth}{!}{
\begin{tabular}{@{}lcccccccc@{}}
\toprule
Model & LR & Eff.\ Batch & Total Steps & Hard Neg. & Warm-up & LR Decay & Attn.\ Drop & MLP Drop \\
\midrule
ColBERT & $1\times10^{-5}$ & 128 & 223,640 & 7 & 5\% & Linear & 0.1 & 0.1 \\
Fallon  & $8\times10^{-5}$ & 128 & 223,640 & 7 & 5\% & Linear & 0.1 & 0.1 \\
\bottomrule
\end{tabular}
}
\end{center}

\subsection{Evaluation Procedure}
\label{sec:evaluation}
We perform retrieval using exhaustive search: for each query, we compute the exact similarity scores for all documents in the corpus and return the top 1,000 results. We use an exact approach over an approximate method, as our new late-interaction method is not directly compatible with standard approximate methods such as PLAID \citep{plaid_late_interaction}. It is worth noting that our method does not preclude an approximate search technique, but developing one was outside the scope of this work.

We evaluate using three evaluation metrics. We use nDCG@10 \citep{ndcg} and P@10 to understand the precision of the approaches and how directly the results could be used in real-world applications where users may be unwilling to go beyond the first set of documents. To get a more holistic view of the retrieval coverage we include Average Precision (AP) which considers both the recall and precision of the retrieved documents. As we retrieve only the top 1,000 documents from a corpus of 100k, AP may have a low upper bound for queries with numerous relevant documents.

\section{Results and Discussion}
\label{sec:results}
The results of the three test datasets for the baseline ColBERT model and our Fallon model are shown in Table \ref{tab:fallon_generalization_results}. The In-Domain results show that the standard MaxSim and ColBERT approach is capable of learning the synthetic task even though it does not naturally handle negative values, which are a natural way to handle negations. Although ColBERT does well, Fallon, which uses Signed MaxSim, is still significantly better across all the measured metrics. 

Moving to the out-of-domain data, we see a much larger gap emerge. For the Different Vocabulary dataset, Fallon actually sees an improvement over the In-Domain dataset (likely due to changes in the number of positives per query), while ColBERT shows a substantial reduction. This result demonstrates that ColBERT's In-Domain performance relies heavily on the shared features between the training and evaluation datasets to work correctly. Fallon, which uses Signed MaxSim, on the other hand, learns a far more general strategy which can directly penalize the presence of unwanted features.

The Negation Only results are the most revealing and provide the strongest evidence for the limitations of standard MaxSim. For Negation Only queries, the majority of the 100k corpus documents are relevant (those lacking the negated attribute), so even random retrieval would do relatively well. Thus, ColBERT's near-zero performance in Table \ref{tab:fallon_generalization_results} indicates it is actively placing documents that contain the negated attribute at the top of the ranking. 

This is a natural consequence of MaxSim. It is very difficult to have a low score for a query embedding, as this requires all document embeddings to have a low score, even those completely unrelated. Because during training ColBERT learns to always match some relevant feature, it struggles when there are no relevant aspects to match. 

Fallon resolves this failure, in large part because Signed MaxSim allows the model to easily penalize features by supporting negative values natively. The relatively low AP, while improved dramatically over ColBERT, is largely an artifact of evaluating only the top 1,000 retrieved documents, since for the Negation Only queries where tens of thousands of documents are relevant, any relevant document outside the top 1,000 contributes zero precision regardless of model quality. The remaining gap from perfect performance suggests that there may still be limitations in generalization, though these limits are far less damaging than the limitations demonstrated by standard MaxSim.

Together, these results across the three evaluation settings demonstrate that Signed MaxSim addresses a fundamental limitation of standard MaxSim: the inability to encode semantic negations and complex relational patterns that require suppression of certain token matches. While ColBERT performs reasonably when training and evaluation distributions are closely aligned, its performance degrades dramatically when the data distribution changes. By leveraging more powerful representations and Signed MaxSim similarity, Fallon is able to generalize far better, making it the clear choice in challenging retrieval scenarios.

\begin{table*}[!t]
\centering
\caption{Performance comparison between the ColBERT baseline and Fallon across in-domain, different-vocabulary, and negation-only evaluation settings. Here, $\dagger$ represents a statistically significant difference between the baseline ColBERT and our model. Statistical significance was calculated using a paired Student’s t-test; we consider any difference with $p < 0.01$ as significant.}
\label{tab:fallon_generalization_results}
\small
\setlength{\tabcolsep}{5.5pt}
\renewcommand{\arraystretch}{1.15}
\resizebox{\textwidth}{!}{
\begin{tabular}{lccc ccc ccc}
\toprule
\multirow{2}{*}{\textbf{Model}}
& \multicolumn{3}{c}{\textbf{In-Domain}}
& \multicolumn{3}{c}{\textbf{Different Vocab.}}
& \multicolumn{3}{c}{\textbf{Negation Only}} \\
\cmidrule(lr){2-4}
\cmidrule(lr){5-7}
\cmidrule(lr){8-10}
& \textbf{nDCG@10} & \textbf{P@10} & \textbf{AP}
& \textbf{nDCG@10} & \textbf{P@10} & \textbf{AP}
& \textbf{nDCG@10} & \textbf{P@10} & \textbf{AP} \\
\midrule
ColBERT
& 0.982 & 0.875 & 0.943
& 0.597 & 0.587 & 0.511
& 0.008 & 0.009 & 0.002 \\

\rowcolor{gray!10}
\textbf{Fallon}
& $\mathbf{0.997^\dagger}$ & $\mathbf{0.893^\dagger}$ & $\mathbf{0.997^\dagger}$
& $\mathbf{1.000^\dagger}$ & $\mathbf{0.955^\dagger}$ & $\mathbf{1.000^\dagger}$
& $\mathbf{0.788^\dagger}$ & $\mathbf{0.789^\dagger}$ & $\mathbf{0.039^\dagger}$ \\
\bottomrule
\end{tabular}
}
\end{table*}

\section{Conclusion}
In this work, we established that MaxSim similarity is at least as expressive as inner-product-based retrieval for non-negative vectors, allowing late-interaction models to exactly replicate infinite-dimensional sparse representations using only $O(k)$ space. We proved that this exact replication is impossible for standard finite-dimensional single-vector models, formally separating the capacity of late-interaction models from single-vector retrievers. However, we also proved that standard MaxSim is incapable of exactly replicating the inner product between real-valued vectors, meaning MaxSim is not universally able to replicate inner product similarity. To address this shortcoming, we leverage our theoretical framework to extend standard MaxSim to enable the exact replication of the inner product between any real-valued vectors. This makes our extension, Signed MaxSim, provably as capable as any single-vector inner-product-based retrieval method. Our experimental results confirm the usefulness of Signed MaxSim and highlight that it is substantially more robust than standard MaxSim when evaluated out-of-domain. Furthermore, we showed that MaxSim naturally evaluates positive Conjunctive Normal Form (CNF) expressions by acting as an aggregation of fuzzy ORs. This finding provides another perspective to explain the empirical improvement seen in late-interaction models and connects late-interaction methods to the traditional IR methods which leverage structured Boolean queries.

Our theoretical results raise several open questions regarding the limits of retrieval similarities. While we show that standard MaxSim cannot replicate inner products between real-valued vectors when the number of embeddings is tied to the sparsity of the original vector, we do not investigate the more relaxed case where the number of embeddings is decoupled from sparsity. Whether such a relaxation would enable exact real-value replication remains an open question. In this work, we show that MaxSim can exactly replicate the inner product between sparse infinite-dimension vectors and that a finite-dimension standard inner product cannot replicate this; however, it is an open question whether all similarities expressible by MaxSim could be exactly expressed or closely approximated by a sparse infinite-dimension inner product. We show that MaxSim can act as a rank-equivalent positive CNF evaluator. Whether it is possible for standard MaxSim to act as a rank-equivalent evaluator of all CNFs remains an open question.

Finally, our findings suggest new directions for neural retrieval architectures. Although our polynomial embedding constructions are primarily theoretical tools, they demonstrate that exact logical evaluation is geometrically possible within late-interaction spaces. Future empirical work could investigate whether trained models implicitly learn these localized structures, or if explicitly regularizing embeddings toward these constructions improves out-of-domain generalization.

\begin{ack} 
The authors would like to thank Vignesh Viswanathan for providing feedback on an early version of this work and Antoine Chaffin for answering questions about ColBERT training.

This work was supported in part by the Center for Intelligent Information Retrieval, in part by the NSF Graduate Research Fellowship Program Award \#1938059, and in part by the Office of Naval Research contract number N000142412612. Any opinions, findings, conclusions or recommendations expressed in this material are those of the authors and do not necessarily reflect those of the sponsor.

The authors used AI technology in various capacities while producing this work, including drafting proofs and other text, checking writing, generating code for experiments, and producing figures and tables. All AI outputs were checked by the authors for correctness, and the authors take full responsibility for the content of this paper. 

\end{ack}

{\small
\bibliographystyle{abbrvnat}
\bibliography{references}
}

\appendix

\end{document}